\documentclass{amsart}
\RequirePackage{url}
\usepackage{verbatim}
\usepackage {   amsfonts    }
\usepackage {   amsmath  }
\usepackage {   amssymb     }
\usepackage {   amsthm      }
\usepackage {   latexsym    }
\usepackage {   graphicx    }
\usepackage {   color       }
\usepackage {   multicol    }

\theoremstyle{plain}
\newtheorem{thm}{Theorem}[section]

\newtheorem{lemma}[thm]{Lemma}
\newtheorem{defn}[thm]{Definition}
\newtheorem{cor}[thm]{Corollary}
\newtheorem{prop}[thm]{Proposition}

\theoremstyle{remark}

\newcommand{\dprod}{\displaystyle\prod}
\newcommand{\norm}[1]{\left|\left|#1\right|\right|}
\newcommand{\innerp}[2]{\left(#1,#2\right)}
\renewcommand{\vec}[1]{\underline{#1}}

\begin{document}
\title[Inverse Inequality Estimates with Symbolic Computation]%
{Inverse Inequality Estimates\\ with\\ Symbolic Computation}

\author[Christoph Koutschan]{Christoph Koutschan}
\address{Christoph Koutschan,
         Johann Radon Institute for Computational and Applied Mathematics (RICAM),
         Austrian Academy of Sciences,
         Altenberger Stra\ss e 69,
         4040 Linz, Austria}
\email{christoph.koutschan@ricam.oeaw.ac.at}
\author[Martin Neum\"uller]{Martin Neum\"uller}
\address{Martin Neum\"uller,
         Institute of Computational Mathematics,
         Johannes Kepler University,
         Altenberger Stra\ss e 69,
         4040 Linz, Austria}
\email{neumueller@numa.uni-linz.ac.at}
\author{Cristian-Silviu Radu}
\address{Cristian-Silviu Radu,
         Research Institute for Symbolic Computation (RISC),
         Johannes Kepler University,
         Altenberger Stra\ss e 69,
         4040 Linz, Austria}
\email{sradu@risc.uni-linz.ac.at}
\thanks{\copyright\ 2016. This manuscript version is made available under the
  CC-BY-NC-ND 4.0 license http:/$\!$/creativecommons.org/licenses/by-nc-nd/4.0/.
  The formal publication of this article is in Advances of Applied Mathematics
  (DOI: 10.1016/j.aam.2016.04.005).}
\thanks{C.~K. was supported by the Austrian Science Fund (FWF): DK W1214.
The research of C.-S.~R. was supported by the strategic program ``Innovatives
O\"O 2010 plus'' by the Upper Austrian Government in the frame of project
W1214-N15-DK6 of the Austrian Science Fund (FWF)}

\begin{abstract}
In the convergence analysis of numerical methods for solving partial
differential equations (such as finite element methods) one arrives at
certain generalized eigenvalue problems, whose maximal eigenvalues need to be
estimated as accurately as possible. We apply symbolic computation methods to
the situation of square elements and are able to improve the previously known
upper bound, given in ``$p$- and $hp$-finite element methods'' (Schwab, 1998),
by a factor of~$8$. More precisely, we try to evaluate the
corresponding determinant using the holonomic ansatz, which is a powerful tool
for dealing with determinants, proposed by Zeilberger in 2007.  However, it
turns out that this method does not succeed on the problem at hand. As a
solution we present a variation of the original holonomic ansatz that is
applicable to a larger class of determinants, including the one we are dealing
with here. We obtain an explicit closed form for the determinant, whose
special form enables us to derive new and tight upper resp. lower bounds on
the maximal eigenvalue, as well as its asymptotic behaviour.
\end{abstract}



\maketitle

\section{Introduction}

Interdisciplinary collaborations between different areas of
mathematics can be hard work because of different terminology and the
difficulty of recognizing the applicability of the methods from one
field in the other field. In Linz there is an almost-20-year tradition
of bringing together researchers from numerical mathematics and
symbolic computation~\cite{Pillwein15,PillweinTakacs14,KoutschanLehrenfeldSchoeberl12,BecirovicEtAl06},
which at the beginning faced exactly these kinds
of problems. Additionally, there could be the risk that the results
are only interesting for one community and not rewarded by the other
one.  Fortunately, this didn't happen in our case: in the current
work, we use and invent tools at the frontier of symbolic computation
to solve a problem that arose at the frontier of numerical analysis
research. Hence, this work improves the knowledge and tools for both
communities.

Inverse inequalities of the form
\begin{align}
 \norm{v_n}_{X(\Omega)} &\leq c_1(h,n) \norm{v_n}_{Y(\Omega)} \qquad \text{for all } v_n \in V_n, \label{equ_generalInverseInequality1}\\
 \norm{v_n}_{Z(\partial\Omega)} &\leq c_2(h,n) \norm{v_n}_{Y(\Omega)} \qquad \text{for all } v_n \in V_n \label{equ_generalInverseInequality2}
\end{align}
play an important role in the analysis and design of numerical methods for
partial differential equations \cite{Ciarlet_1978, Schwab_1998, Brenner_2008,
  Pietro_2012} and in the construction of efficient solvers for the arising
linear systems of those methods \cite{Hackbusch_1985,
  Trottenberg_2001}. Bounds for the constants of the type
\eqref{equ_generalInverseInequality1}--\eqref{equ_generalInverseInequality2}
have been studied for example in \cite{Schwab_1998, Warburton_2003,
  Verfuerth_2004, Graham_2005, Georgoulis_2008}, where the asymptotic
behaviour with respect to $h$ and $n$ is covered but usually these constants
are over estimated. In many numerical methods a precise knowledge of these
constants is required, which motivates this work where we use and present
tools from symbolic computation to derive precise estimates.

Here $\Omega \subset \mathbb{R}^d, d \in \mathbb{N}$, is a bounded and open
set with sufficiently smooth boundary $\partial\Omega$, describing a finite
element with diameter $h > 0$ which is used in the numerical method (intervals
for $d=1$, often triangles or quadrilaterals for $d=2$, usually tetrahedra or
hexahedra for $d=3$, \ldots). Let $V$ be some infinite-dimensional space of
functions defined on $\Omega$ such that the solution of the PDE is an element
of~$V$.  With $(V_n)_{n \in \mathbb{N}}$ we denote a family of
finite-dimensional (usually closed) subspaces of $V$ whose dimension depends
on~$n$; the desired solution of the PDE is approximated by an element
of~$V_n$.  Moreover we have some given norms $\norm{\cdot}_{X(\Omega)}$,
$\norm{\cdot}_{Y(\Omega)}$ and $\norm{\cdot}_{Z(\partial\Omega)}$ which are
induced by certain inner products $\innerp{\cdot}{\cdot}_{X(\Omega)}$,
$\innerp{\cdot}{\cdot}_{Y(\Omega)}$ and
$\innerp{\cdot}{\cdot}_{Z(\partial\Omega)}$, which are used in the analysis of
the numerical methods. In general the constants $c_1$ and $c_2$
of~\eqref{equ_generalInverseInequality1}
and~\eqref{equ_generalInverseInequality2} depend on the diameter~$h$ and on
the parameter~$n$ reflecting the dimension of the space~$V_n$. The dependence
with respect to the diameter~$h$ is obtained by transforming Equations
\eqref{equ_generalInverseInequality1} and
\eqref{equ_generalInverseInequality2} to a reference domain $\hat{\Omega}
\subset \mathbb{R}^d$, i.e.
\begin{align}
 \norm{\hat{v}_n}_{X(\hat{\Omega})} &\leq \hat{c}_1(n) \norm{\hat{v}_n}_{Y(\hat{\Omega})} \qquad \text{for all } \hat{v}_n \in \hat{V}_n, \label{equ_referenceInverseInequality1}\\
 \norm{\hat{v}_n}_{Z(\partial\hat{\Omega})} &\leq \hat{c}_2(n) \norm{\hat{v}_n}_{Y(\hat{\Omega})} \qquad \text{for all } \hat{v}_n \in \hat{V}_n \label{equ_referenceInverseInequality2}
\end{align}
and applying a scaling argument \cite{Schwab_1998, Warburton_2003,
  Brenner_2008}. The more challenging problem is to find precise estimates for
the constants $\hat{c}_1$ and $\hat{c}_2$ with respect to the
parameter~$n$. The best possible constants by definition are given by
\begin{align}
 \hat{c}_1(n) &= \sup_{\hat{v}_n \in \hat{V}_n} \frac{\norm{\hat{v}_n}_{X(\hat{\Omega})}}{\norm{\hat{v}_n}_{Y(\hat{\Omega})}} = \sqrt{ \sup_{\hat{v}_n \in \hat{V}_n} \frac{\innerp{\hat{v}_n}{\hat{v}_n}_{X(\hat{\Omega})}}{\innerp{\hat{v}_n}{\hat{v}_n}_{Y(\hat{\Omega})}} }, \label{equ_constant1} \\
  \hat{c}_2(n) &= \sup_{\hat{v}_n \in \hat{V}_n} \frac{\norm{\hat{v}_n}_{Z(\partial\hat{\Omega})}}{\norm{\hat{v}_n}_{Y(\hat{\Omega})}} = \sqrt{ \sup_{\hat{v}_n \in \hat{V}_n} \frac{\innerp{\hat{v}_n}{\hat{v}_n}_{Z(\hat{\partial\Omega})}}{\innerp{\hat{v}_n}{\hat{v}_n}_{Y(\hat{\Omega})}} }. \label{equ_constant2}
\end{align}
Introducing for $\hat{V}_n$ the basis functions $(\varphi_k)_{1\leq k\leq n}$,
i.e., $\hat{V}_n = \mathrm{span}\{ \varphi_1,\dots,\varphi_n \}$, we further
obtain
\[
 \bigl(\hat{c}_1(n)\bigr)^2 = \sup_{\vec{v}_n \in \mathbb{R}^n} \frac{\innerp{K_n \vec{v}_n}{\vec{v}_n}_{\ell^2}}{\innerp{M_n \vec{v}_n}{\vec{v}_n}_{\ell^2}} \quad \text{and} \quad \bigl(\hat{c}_2(n)\bigr)^2 = \sup_{\vec{v}_n \in \mathbb{R}^n} \frac{\innerp{L_n \vec{v}_n}{\vec{v}_n}_{\ell^2}}{\innerp{M_n \vec{v}_n}{\vec{v}_n}_{\ell^2}},
\]
with the symmetric and positive (semi-) definite matrices
\[
  K_n(i,j) := \innerp{\varphi_j}{\varphi_i}_{X(\hat{\Omega})},\ 
  M_n(i,j) := \innerp{\varphi_j}{\varphi_i}_{Y(\hat{\Omega})},\ \text{and}\ 
  L_n(i,j) := \innerp{\varphi_j}{\varphi_i}_{Z(\hat{\partial\Omega})}
\]
for $i,j = 1,\ldots,n$. Note that the rows of the matrices are naturally defined via the second argument of the inner product and the columns are given naturally via the first argument of the inner product. Hence, the constants $(\hat{c}_1)^2$ and
$(\hat{c}_2)^2$ are given by the largest eigenvalues of the generalized
eigenvalue problems
\begin{align}
\label{eq.ev}
  K_n \vec{x}_n = \lambda_n M_n \vec{x}_n \qquad \text{and} \qquad L_n \vec{x}_n = \mu_n M_n \vec{x}_n,
\end{align}
i.e.
\begin{align*}
 \bigl(\hat{c}_1(n)\bigr)^2 = \lambda_n \qquad \text{and} \qquad \bigl(\hat{c}_2(n)\bigr)^2 = \mu_n.
\end{align*}

In this work we want to solve the
problems~\eqref{equ_referenceInverseInequality1}
and~\eqref{equ_referenceInverseInequality2}, i.e., determine estimates for
$\hat{c}_1(n)$ and $\hat{c}_2(n)$, for the reference domain $\hat{\Omega} =
(-1,1)^2$ with
\begin{align*}
  \innerp{u}{v}_{X(\hat{\Omega})} &= \int_{\hat{\Omega}} \partial_x u(x,y) \partial_x v(x,y) \,\mathrm{d}x\,\mathrm{d}y,\\
  \innerp{u}{v}_{Y(\hat{\Omega})} &= \int_{\hat{\Omega}} u(x,y) v(x,y) \,\mathrm{d}x\,\mathrm{d}y,
\end{align*}
for $u,v \in \hat{V}_n$, where $\hat{V}_n$ is the space of polynomials of degree less than~$n$, i.e.
\begin{align*}
 \hat{V}_n = \left\{ x^i y^j : 0 \leq i,j < n \right\}.
\end{align*}
In Section~\ref{section_evEstimates} we
state the problem in detail and derive its formulation as a generalized
eigenvalue problem of the form~\eqref{eq.ev}.  The difficulty is now to find
an accurate estimate for the largest eigenvalues $\lambda_n$ and $\mu_n$ for
general parameter $n \in \mathbb{N}$. Of course one can compute the
eigenvalues exactly for a given fixed parameter~$n$, which is done for example
in \cite{Ozisik_2010}. But to derive their exact values or precise estimates
for a general parameter~$n$ one needs techniques from symbolic computation. In
Section~\ref{section_detEvaluation} we use the HolonomicFunctions
package~\cite{Koutschan09,Koutschan10b} to prove a closed-form representation
of the characteristic polynomial of our eigenvalue problem, in the spirit of
the holonomic ansatz~\cite{Zeilberger07} for evaluating determinants. The
holonomic ansatz is a very powerful method (it was the key to the ``holy grail
of enumerative combinatorics''~\cite{KoutschanKauersZeilberger11}), and a very
flexible one, too~\cite{KoutschanThanatipanonda13}. For our purposes here we
had to adapt this algorithm, which led to a new variant that is applicable to
a larger class of determinants. In Section~\ref{section_bounds}, our
representation of the characteristic polynomial is used to derive and prove
the estimates from below and above
\begin{align*}
  \frac{1}{4} \sqrt{n (n-1) (n+1)(n+2)} \leq \hat{c}_1(n) \leq \frac{1}{2 \sqrt{2}}\sqrt{n (n-1) (n+1)(n+2)}
\end{align*}
for the constant~$\hat{c}_1(n)$ in the inequality
\begin{align*}
 \norm{\partial_{\hat{x}}\hat{u}_n}_{L^2(\hat{\Omega})} \leq
 \hat{c}_1(n) \norm{\hat{u}_n}_{L^2(\hat{\Omega})} \qquad \text{for all } \hat{u}_n \in \hat{V}_n.
\end{align*}
In Lemmas \ref{lowerbound} and \ref{upperbound} we give much sharper estimates
for~$\hat{c}_1(n)$; as an application, they allow to tune the parameters of
numerical methods precisely.  In Section~\ref{section_asympt} the same
representation is used to investigate the asymptotic behaviour of the
eigenvalues; on the way we discover some interesting connections to the Taylor
expansions of trigonometric functions.  As an encore, we deal with the second
inequality~\eqref{equ_referenceInverseInequality2} in
Section~\ref{section_easydet}.  It turns out that it is considerably simpler
and we are able to derive the exact value of $\hat{c}_2(n)$.

Throughout the paper, we employ the following notation:
$(a)_n$ denotes the Poch\-hammer symbol, also
known as rising factorial, defined for all nonnegative integers~$n$ by
\[
  (a)_n:=a\cdot(a+1)\cdots(a+n-1)\text{ for }n>0\quad\text{and}\quad(a)_0:=1.
\]
We use $\lfloor x\rfloor$ for the floor function, and $\lceil x\rceil$ for the
ceiling function, i.e., the largest integer below~$x$, resp.  the smallest
integer above~$x$. For a polynomial~$p$ we refer to the degree of~$p$ with
respect to the variable~$x$ by $\deg_x(p)$. By $\delta_{i,j}$ we denote the
Kronecker delta symbol, i.e., $\delta_{i,j}=0$ if $i\neq j$ and
$\delta_{i,i}=1$.  If $A$ is the $n\times n$ matrix $(a_{i,j})_{1\leq i,j\leq
  n}$, then we use for its determinant the short-hand notation
$\det(A)=\det_{1\leq i,j\leq n}(a_{i,j})$. The determinant of the $0\times0$
matrix is defined to be~$1$.

\section{The Maximal Eigenvalue Problem}
\label{section_evEstimates}

Let the reference domain $\hat{\Omega}\subset \mathbb{R}^2$ be defined by
$\hat{\Omega}:=(-1,1)^2$, the open square of size~$2$ centered around the
origin.  For $n\in\mathbb{N}$ and $k\in\{1,\dots,n^2\}$, we define $\chi_n(k)$
and $\rho_n(k)$ to be the unique integers in $\{0,\dots,n-1\}$ satisfying
$k=\chi_n(k)\cdot n+\rho_n(k)+1$.  In other words,
\[
  \chi_n(k) := \Bigl\lfloor\frac{k-1}{n}\Bigr\rfloor
  \qquad\text{and}\qquad
  \rho_n(k) := k-1\mathrel{\mathrm{mod}}n.
\]
For the rest of this section we fix $n\in\mathbb{N}$ and write shortly
$\chi(k)$ and $\rho(k)$.  By employing the standard monomial basis
$\varphi_k:=x^{\rho(k)}t^{\chi(k)}$, we obtain the $n^2\times n^2$ matrix
$M_n$ with entries $m_{i,j}$ defined by
\begin{equation}\label{eq.M}
  m_{i,j}:=\int_{\hat{\Omega}} \varphi_i \varphi_j \,\mathrm{d}x\,\mathrm{d}t \qquad (1\leq i,j\leq n^2)
\end{equation}
and the $n^2\times n^2$ matrix $K_n$ with entries
\begin{equation}\label{eq.K}
  k_{i,j}:=\int_{\hat{\Omega}}(\partial_x \varphi_i)(\partial_x\varphi_j) \,\mathrm{d}x\,\mathrm{d}t \qquad (1\leq i,j\leq n^2).
\end{equation}
Since $\varphi_i$ and $\varphi_j$ are just monomials, these integrals can be
evaluated in a straightforward manner:
\begin{align*}
m_{i,j} &= \int_{-1}^{1}\left(\int_{-1}^1 x^{\rho(i)}t^{\chi(i)} x^{\rho(j)}t^{\chi(j)} \,\mathrm{d}x\right)\mathrm{d}t\\
&= \int_{-1}^1\frac{1-(-1)^{\rho(i)+\rho(j)+1}}{\rho(i)+\rho(j)+1} \, t^{\chi(i)+\chi(j)} \,\mathrm{d}t\\
&= \frac{1-(-1)^{\rho(i)+\rho(j)+1}}{\rho(i)+\rho(j)+1}\cdot\frac{1-(-1)^{\chi(i)+\chi(j)+1}}{\chi(i)+\chi(j)+1}.
\end{align*}
Similarly
\begin{align*}
k_{i,j} &= \int_{-1}^1\left(\int_{-1}^1 \rho(i)\rho(j)x^{\rho(i)-1}t^{\chi(i)} x^{\rho(j)-1}t^{\chi(j)} \,\mathrm{d}x\right)\mathrm{d}t\\
&= \int_{-1}^1\rho(i)\rho(j)\frac{1-(-1)^{\rho(i)+\rho(j)-1}}{\rho(i)+\rho(j)-1}\,t^{\chi(i)+\chi(j)} \,\mathrm{d}t\\
&= \rho(i)\rho(j)\,\frac{1-(-1)^{\rho(i)+\rho(j)-1}}{\rho(i)+\rho(j)-1}\cdot\frac{1-(-1)^{\chi(i)+\chi(j)+1}}{\chi(i)+\chi(j)+1},
\end{align*}
where we assumed that $\rho(i)+\rho(j)>1$; otherwise the integral equals to~$0$.

We are interested in computing the maximal $\lambda_{n}\in\mathbb{R}$ such
that $\det(K_{n}-\lambda_{n}M_{n})=0$. In the following we derive an
equivalent formulation of this problem that involves smaller matrices. For
this purpose let 
\[
  a_{i,j}:=\frac{1-(-1)^{i+j-1}}{i+j-1}
  \quad\text{and}\quad
  b_{i,j}:=(i-1)(j-1)\frac{1-(-1)^{i+j-3}}{i+j-3},
\]
such that the matrix entries $m_{i,j}$ and $k_{i,j}$ can be written as
\begin{align*}
  m_{i,j} &= a_{\chi(i)+1,\chi(j)+1} \cdot a_{\rho(i)+1,\rho(j)+1} \\
  k_{i,j} &= a_{\chi(i)+1,\chi(j)+1} \cdot b_{\rho(i)+1,\rho(j)+1}.
\end{align*}
This shows that the matrices $M_n$ and $K_n$ can be written as Kronecker products:
\[
  M_n = A_n\otimes A_n
  \quad\text{and}\quad
  K_n = A_n\otimes B_n.
\]
These representations as Kronecker products are quite natural since the used basis functions and the reference domain $\hat{\Omega}$ itself have tensor product structure.
In particular we then obtain,
\[
  \det(K_n-\lambda_nM_n) = \det\bigl(A_n \otimes (B_n-\lambda_nA_n)\bigr)
  = \det(A_n)^n \det(B_n-\lambda_nA_n)^n.
\]
So the problem is equivalent to computing the maximal $\lambda_{n}\in\mathbb{R}$ such that
\[
  \det(B_{n}-\lambda_{n}A_{n})=0.
\]

\section{Determinant Evaluation}
\label{section_detEvaluation}

According to the previous discussion, we are now interested in evaluating the
determinant
\[
  \det(B_{n}-\lambda A_{n}) =
  \det_{1\leq i,j\leq n}\biggl( \bigl(1-(-1)^{i+j-1}\bigr)\Bigl(\frac{(i-1)(j-1)}{i+j-3}-\frac{\lambda}{i+j-1}\Bigr)\biggr)
\]
for symbolic~$\lambda$; the desired maximal eigenvalue~$\lambda_n$ is then
just the largest root of the obtained polynomial.  We see that the matrix
$B_{n}-\lambda A_{n}$ has zeros at all positions $(i,j)$ for which $i+j$ is an
odd integer. By applying the permutation $(2,4,6,\dots,1,3,5,\dots)$ to the
rows and to the columns of the matrix, we decompose it into block form and
obtain
\[
  \det(B_n-\lambda A_n) =
  2^n \begin{vmatrix}A^{(0)}_{\lfloor n/2\rfloor} & 0\\ 0& A^{(1)}_{\lceil n/2\rceil}\end{vmatrix} =
  2^n \det\Bigl(A^{(0)}_{\lfloor n/2\rfloor}\Bigr) \cdot \det\Bigl(A^{(1)}_{\lceil n/2\rceil}\Bigr)
\]
where the subscripts indicate the dimensions of the square matrices $A^{(0)}$ and $A^{(1)}$,
whose entries are independent of the dimension and given by
\begin{align}
  \label{eq.a0} a^{(0)}_{i,j} &:= \frac{(2i-1)(2j-1)}{2i+2j-3} - \frac{\lambda}{2i+2j-1}, \\
  \label{eq.a1} a^{(1)}_{i,j} &:= \frac{4(i-1)(j-1)}{2i+2j-5} - \frac{\lambda}{2i+2j-3}.
\end{align}
Hence the matrices $A^{(0)}$ and $A^{(1)}$ start as follows:
\begin{align*}
  A^{(0)} &= \begin{pmatrix}
  1-\frac{\lambda}{3} & 1-\frac{\lambda}{5} & 1-\frac{\lambda}{7} & 1-\frac{\lambda}{9} & \cdots \\[1ex]
  1-\frac{\lambda}{5} & \frac{9}{5}-\frac{\lambda}{7} & \frac{15}{7}-\frac{\lambda}{9} & \frac{7}{3}-\frac{\lambda}{11} & \cdots \\[1ex]
  1-\frac{\lambda}{7} & \frac{15}{7}-\frac{\lambda}{9} & \frac{25}{9}-\frac{\lambda}{11} & \frac{35}{11}-\frac{\lambda}{13} & \cdots \\[1ex]
  1-\frac{\lambda}{9} & \frac{7}{3}-\frac{\lambda}{11} & \frac{35}{11}-\frac{\lambda}{13} & \frac{49}{13}-\frac{\lambda}{15} & \cdots \\
  \vdots & \vdots & \vdots & \vdots & \ddots
  \end{pmatrix},\\
  A^{(1)} &= \begin{pmatrix}
  -\lambda & -\frac{\lambda}{3} & -\frac{\lambda}{5} & -\frac{\lambda}{7} & \cdots \\[1ex]
  -\frac{\lambda}{3} & \frac{4}{3}-\frac{\lambda}{5} & \frac{8}{5}-\frac{\lambda}{7} & \frac{12}{7}-\frac{\lambda}{9} & \cdots \\[1ex]
  -\frac{\lambda}{5} & \frac{8}{5}-\frac{\lambda}{7} & \frac{16}{7}-\frac{\lambda}{9} & \frac{8}{3}-\frac{\lambda}{11} & \cdots \\[1ex]
  -\frac{\lambda}{7} & \frac{12}{7}-\frac{\lambda}{9} & \frac{8}{3}-\frac{\lambda}{11} & \frac{36}{11}-\frac{\lambda}{13} & \cdots \\
  \vdots & \vdots & \vdots & \vdots & \ddots
  \end{pmatrix}.
\end{align*}

\begin{thm}
\label{thm.detA0A1}
Let $a^{(0)}_{i,j}$ and $a^{(1)}_{i,j}$ be defined as in \eqref{eq.a0}
and~\eqref{eq.a1}, then the following identities hold for all nonnegative
integers~$n$:
\begin{align*}
  \det A^{(0)}_n &= \det_{1\leq i,j\leq n} a^{(0)}_{i,j} = (-1)^n \, h^{(0)}_n \cdot F_{2n}(\lambda),\\
  \det A^{(1)}_n &= \det_{1\leq i,j\leq n} a^{(1)}_{i,j} = (-1)^n \, h^{(1)}_n \cdot \lambda F_{2n-1}(\lambda),
\end{align*}
where
\begin{align}
  \label{eq.f}
  F_n(\lambda) &:= \sum_{j=0}^{\nu} (-4)^{j-\nu} \frac{(2\nu-2j+1)_{n}}{(2j-2\nu+n)!} \lambda^j
  \quad\text{with}\quad \nu = \nu(n) := \left\lfloor\frac{n}{2}\right\rfloor, \\
  \label{eq.h}
  h^{(\ell)}_n &:= \frac{1}{2^n} \prod_{i=1}^n \frac{\bigl((i-1)!\bigr){}^2}{\left(i-\ell+\frac12\right)_n}.
\end{align}
\end{thm}

\begin{cor}
For all nonnegative integers~$n$ we have
\[
  \det(B_{n}-\lambda A_{n}) =
  (-2)^n \, h^{(0)}_{\lfloor n/2\rfloor} \, h^{(1)}_{\lceil n/2\rceil} \, \lambda \, F_{n-1}(\lambda) \, F_{n}(\lambda).
\]
\end{cor}

The key ingredient for the proof of Theorem~\ref{thm.detA0A1} is the following
lemma which shows that the quantities $p^{(0)}_{n,j}$ and $p^{(1)}_{n,j}$
defined there are basically the entries of the last column of the inverses of
$A^{(0)}$ and $A^{(1)}$, respectively.

\begin{lemma}
\label{lem.invA0A1}
With
\begin{align*}
  p^{(0)}_{n,j} &=
    \frac{2^{2n+2j-3} \left(\frac32\right)_{2n-1} \left(n+\frac12\right)_{j-1}}{(n-1)! \, (2j-1)!}
    \sum_{m=0}^{n-1} \sum_{k=0}^{2n-2m-2} \frac{(-1)^{j+m} \, (2m+1)_{2k} \, \lambda^m}{4^{m+k} k! \, (2m+k-n-j+2)!},\\
  p^{(1)}_{n,j} &=
    \frac{4^{j-n} (4n-3)! \left(n-\frac12\right)_{j-1}}{(2 n-2)! \, (n-1)! \, (2 j-2)!}
    \sum_{m=0}^{n-1} \sum_{k=0}^{2n-2m-2} \frac{(-1)^{j+m} \, (2m)_{2k} \, \lambda^m}{4^{m+k} k! \, (2m+k-n-j+2)!},
\end{align*}
the following identities hold for all nonnegative integers~$n$ and for $1\leq i\leq n$:
\begin{align*}
  \sum_{j=1}^n a^{(0)}_{i,j} p^{(0)}_{n,j} &= \delta_{i,n} F_{2n}(\lambda),\\
  \sum_{j=1}^n a^{(1)}_{i,j} p^{(1)}_{n,j} &= \delta_{i,n} \lambda F_{2n-1}(\lambda).
\end{align*}
\end{lemma}
\begin{proof}
These identities can be proven routinely using the holonomic systems
approach~\cite{Zeilberger90}. We have carried out the necessary calculations
using the HolonomicFunctions package~\cite{Koutschan09,Koutschan10b}. The
results are documented in the supplementary electronic material~\cite{elec}.

First we derive, using holonomic closure properties and creative telescoping,
a (left Gr\"obner) basis for the set of recurrence equations that
$p^{(0)}_{n,j}$ satisfies. Again applying closure properties (in this case for
multiplication) one obtains recurrences for the product $a^{(0)}_{i,j}
p^{(0)}_{n,j}$, and by creative telescoping, for its definite sum, which we
denote by $l_{i,n}$ (it is the left-hand side of the first identity).  Here we
face the problem of poles inside the summation range that are introduced by
the certificate of the telescopic relation. We solve this issue by
constructing a different certificate, free of the problematic denominators,
using an ansatz reminiscent of the polynomial
ansatz~\cite[Sec. 3.4]{Koutschan09}.  The recurrences for $l_{i,n}$ have the
following form (some polynomial coefficients are omitted for space reasons):
\begin{align*}
  & 16 n^4 (n+1)^2 (2 n+1)^4 (2 n+3)^2 (4 n+1) (i-n+1)^2 (2 i+2 n+3)^2 \\
    &\qquad \times\left(4 i^2+2 i+\lambda -4 n^2-2 n\right)^2 l_{i,n+2} = (\cdots)l_{i+1,n}+(\cdots)l_{i,n+1}+(\cdots)l_{i,n}, \\[1ex]
  & 2 n (2 n+1) (i-n+1) (2 i+2 n+3) \left(4 i^2+2 i+\lambda -4 n^2-2 n\right) l_{i+1,n+1} = {}\\
    &\qquad (\cdots)l_{i+1,n} + (\cdots)l_{i,n+1} + (\cdots)l_{i,n}, \\[1ex]
  & 2 (n-1) n (2 n-1) (2 n+1) (4 n+1)^2 (4 n+3) (i-n+1)^2 (i-n+2) (2 i+2 n+3) \\
    &\qquad \times\left(4 i^2+2 i+\lambda -4 n^2-2 n\right) l_{i+2,n} = (\cdots)l_{i+1,n} + (\cdots)l_{i,n+1} + (\cdots)l_{i,n}.
\end{align*}
From their support and their leading coefficients it becomes clear that when
we want to use them to compute $l_{i,n}$ for all $1\leq i<n$, then we have to
give the initial conditions $l_{1,2}$, $l_{1,3}$, $l_{1,4}$, and $l_{2,3}$.
By verifying that they all equal~$0$ we have shown that the first identity
holds for $i<n$.

For $i=n$ we can construct, by holonomic substitution, a univariate recurrence
satisfied by $l_{n,n}$. It turns out that the corresponding operator is a left
multiple of the second-order operator that annihilates $F_{2n}$. Also in this
case, the proof can be completed by checking a few initial conditions. The
proof of the second identity is established in an analogous way.
\end{proof}

\begin{lemma}
\label{lem.detLC}
The following determinant evaluations hold for all nonnegative integers~$n$:
\begin{align*}
  \det_{1\leq i,j\leq n} \left(\frac{1}{2i+2j-1}\right) &=
    \frac{1}{2^n} \prod_{i=1}^n \frac{\bigl((i-1)!\bigr){}^2}{\left(i+\frac12\right)_n}
    \quad \Bigl(= h^{(0)}_n\Bigr), \\
  \det_{1\leq i,j\leq n} \left(\frac{1}{2i+2j-3}\right) &=
    \frac{1}{2^n} \prod_{i=1}^n \frac{\bigl((i-1)!\bigr){}^2}{\left(i-\frac12\right)_n}
    \quad \Bigl( = h^{(1)}_n\Bigr).
\end{align*}
\end{lemma}
\begin{proof}
These determinants are special cases of Cauchy's classic double
alternant~\cite{Cauchy41}
\[
  \det_{1\leq i,j\leq n} \left(\frac{1}{x_i+y_j}\right) =
  \frac{\dprod_{1\leq i<j\leq n} (x_i-x_j)(y_i-y_j)}{\dprod_{1\leq i,j\leq n} (x_i+y_j)}
\]
where $x_1,\dots,x_n,y_1,\dots,y_n$ are indeterminates;
see also \cite[Thm. 12, Eq. (5.5)]{Krattenthaler05} and
\cite[Thm. 15]{Kuperberg02} for a proof by factor exhaustion.
In order to obtain the first assertion, we specialize $x_k=2k$ and
$y_k=2k-1$ for $1\leq k\leq n$, and obtain:
\begin{align*}
  \det_{1\leq i,j\leq n} \left(\frac{1}{2i+2j-1}\right) &=
  \frac{\dprod_{1\leq i<j\leq n} \!\! (2i-2j)^2}{\dprod_{1\leq i,j\leq n} \!\! (2i+2j-1)} =
  \frac{\dprod_{i=1}^n \bigl(2^{n-i}(n-i)!\bigr){}^2}{\dprod_{i=1}^n 2^n\,\bigl(i+\tfrac12\bigr)_n} \\
  &= \frac{1}{2^n} \prod_{i=1}^n \frac{\bigl((i-1)!\bigr){}^2}{\bigl(i+\tfrac12\bigr)_n}.
\end{align*}
The second assertion is derived in a completely analogous way.
Note also that these two determinants can be proven routinely
using the holonomic ansatz~\cite{Zeilberger07}.
\end{proof}

\begin{proof}[Proof of Theorem~\ref{thm.detA0A1}]
Lemma~\ref{lem.invA0A1} shows that the vector
$\bigl(p^{(\ell)}_{n,1},\dots,p^{(\ell)}_{n,\vphantom{1}n}\bigr)^T$ is, up to
a scalar multiple, the $n$-th column of $\bigl(A^{(\ell)}_n\bigr)^{-1}$ for $\ell=0,1$. Since
the entries of this vector (and of course those of the matrices $A^{(\ell)}$
itself) are polynomials in~$\lambda$, this shows that
$\det\bigl(A^{(0)}_n\bigr) \mid F_{2n}(\lambda)$ and that
$\det\bigl(A^{(1)}_n\bigr) \mid \lambda F_{2n-1}(\lambda)$.  Note that both
polynomials $F_{2n}(\lambda)$ and $\lambda F_{2n-1}(\lambda)$ have
degree~$n$ in~$\lambda$.  Next we argue that also the determinants of
$A^{(0)}_n$ and $A^{(1)}_n$ have degree~$n$ in~$\lambda$, which is the maximal
possible---taking into account that the matrix entries are linear polynomials
in~$\lambda$.  Observe that the matrix entries in Lemma~\ref{lem.detLC} are
precisely $\lim_{\lambda\to\infty} -a^{(\ell)}_{i,j}/\lambda$.  Thus
Lemma~\ref{lem.detLC} implies that $\det\bigl(A^{(\ell)}_n/\lambda\bigr) =
\lambda^{-n}\det\bigl(A^{(\ell)}_n\bigr)$ converges to a nonzero constant
(only depending on~$n$) as $\lambda$ goes to infinity. Hence
$\deg_{\lambda}\bigl(\det A^{(\ell)}_n\bigr)=n$ for $\ell=0,1$, which means
that the two determinants are now determined up to a multiplicative constant
not depending on~$\lambda$. By noting that the polynomials $F_n(\lambda)$ are
monic and that the expressions given in Lemma~\ref{lem.detLC} are, up to sign,
the leading coefficients of $\det\bigl(A^{(0)}_n\bigr)$ and
$\det\bigl(A^{(1)}_n\bigr)$, respectively, the assertion of the theorem is
proven.
\end{proof}

Note that our proof of the determinant evaluations in Theorem~\ref{thm.detA0A1}
is very reminiscent of Zeilberger's holonomic ansatz~\cite{Zeilberger07}.
In fact, the only difference is that we chose to normalize the vector
$v_n=\bigl(p^{(0)}_{n,1},\dots,p^{(0)}_{n,\vphantom{1}n}\bigr)^T$ in a different
way as Zeilberger would do it: while he suggests the normalization $p^{(0)}_{n,n}=1$,
we normalize $v_n$ such that $A^{(0)}_nv_n=\bigl(0,\dots,0,q_n(\lambda)\bigr)^T$ and
$q_n(\lambda)$ is a monic polynomial with $\deg_{\lambda}(q_n)=n$. (The same
discussion applies to $A^{(1)}_n$, of course.)

In the original formulation of the holonomic ansatz, i.e., with the
normalization $p^{(0)}_{n,n}=1$, the final result in the case of success is a
holonomic recurrence, i.e., a linear recurrence with polynomial coefficients,
for $\det\bigl(A^{(0)}_{n+1}\bigr)/\det\bigl(A^{(0)}_n\bigr)$.  However, this
ansatz is not at all guaranteed to succeed: even if the matrix entries are
holonomic, this doesn't mean that the sequence of quotients of consecutive
determinants is a holonomic sequence. The determinant of $A^{(0)}_n$ is such
an example: the polynomials $\bigl(F_{2n}(\lambda)\bigr)_{n\geq1}$ satisfy the
second-order recurrence
\begin{multline*}
  (4n+3) F_{2n+4}(\lambda) + (4n+5)(16n^2+40n-2\lambda+21) F_{2n+2}(\lambda) \\
  +(4n+7) \lambda^2 F_{2n}(\lambda) = 0,
\end{multline*}
which means that (most likely) the quotient
$F_{2n+2}(\lambda)/F_{2n}(\lambda)$ doesn't satisfy a holonomic recurrence
of any order. (We have strong evidence that this quotient is non-holonomic,
but we haven't tried to prove this rigorously.) Provided that this is true,
the original holonomic ansatz must fail.

Thanks to the additional parameter $\lambda$ that appears polynomially in the
matrix entries, we can identify the determinant of $A^{(0)}_n$ in the
denominators of the inverse matrix. Thus a natural normalization of the
vector~$v_n$ would be such that $A^{(0)}_nv_n=\bigl(0,\dots,0,\det
A^{(0)}_n\bigr)^T$. In that case, the final result would be a holonomic
recurrence for $\det A^{(0)}_n$; hence this variant is applicable when the
determinant itself is a holonomic sequence in~$n$. Unfortunately, that's not
the case for the matrix $A^{(0)}_n$ because of the non-holonomic
prefactor~$h^{(0)}_n$. This explains why we had to choose yet another
normalization, in order to separate the holonomic and the non-holonomic part
of the determinant. For each part then we had to prove a different determinant
evaluation: for the holonomic ``polynomial part'' this was done in
Lemma~\ref{lem.invA0A1}, for the non-holonomic ``constant part'' in
Lemma~\ref{lem.detLC}. It is not unlikely that there are many more examples
of determinants where the original holonomic ansatz fails, but where the
modifications described here lead to success.

At the end of this section we want to briefly discuss an alternative way to
derive the polynomials $F_n(\lambda)$. In our above considerations we started
with the monomial basis when formulating the eigenvalue problem. Alternatively,
one could employ the Legendre basis leading to the following determinant:
\[
  D_n = \det_{1\leq i,j\leq n} \biggl(\int_{-1}^1 P'_i(x)P'_j(x)\,\mathrm{d}x-\lambda\int_{-1}^1 P_i(x)P_j(x)\,\mathrm{d}x\biggr)
\]
(note that only the matrix entries on the main diagonal depend
on~$\lambda$). Indeed any basis $(\varphi_k)_{1\leq k\leq n}$ for the space
$\hat{V}_n$ can be used for the computation of the eigenvalues given in
\eqref{eq.ev}. So by construction, this determinant leads to the same family
of polynomials $F_n(\lambda)$, and in fact we have that
$\lambda\det(B_n-\lambda A_n)/D_{n+1}$ does not depend on~$\lambda$. Doing the
same block decomposition as before, we obtain the two families of matrices
\[
  \biggl(2m(2m+1)-\delta_{i,j}\frac{2\lambda}{4i+1}\biggr)_{1\leq i,j\leq n}
  \qquad\text{and}\qquad
  \biggl(2m(2m-1)-\delta_{i,j}\frac{2\lambda}{4i-1}\biggr)_{1\leq i,j\leq n}
\]
where $m$ stands for $\min(i,j)$, whose determinants are given by
\[
  \frac{(-1)^n}{2^n\left(\frac54\right)_n}F_{2n+1}(\lambda)
  \quad\text{resp.}\quad
  \frac{(-1)^n}{2^n\left(\frac34\right)_n}F_{2n}(\lambda).
\]
Note that these determinants are ``nicer'' than the ones we considered above,
because their leading coefficients form holonomic sequences (actually they are
hypergeometric). So it seems that we should have started with this formulation.
But there is also a drawback: the matrix entries are defined in terms of $\min(i,j)$,
which on the one hand yields nicely structured matrices (constant along ``hooks'',
with a perturbation on the diagonal) such as
\[
  \begin{pmatrix}
   2-\frac{2 x}{3} & 2 & 2 & 2 & 2 & \dots \\
   2 & 12-\frac{2 x}{7} & 12 & 12 & 12 & \dots  \\
   2 & 12 & 30-\frac{2 x}{11} & 30 & 30 & \dots  \\
   2 & 12 & 30 & 56-\frac{2 x}{15} & 56 & \dots  \\
   2 & 12 & 30 & 56 & 90-\frac{2 x}{19} & \dots  \\
   \vdots & \vdots & \vdots & \vdots & \vdots & \ddots
   \end{pmatrix},
\]
but on the other hand requires case distinctions that make the proofs of the
relevant identities (the analog of Lemma~\ref{lem.invA0A1}) more complicated.

\section{Upper and Lower Bounds on the Maximal Root of $F_n(\lambda)$}
\label{section_bounds}
In this section we give lower and upper bounds on the maximal root
of~$F_n(\lambda)$. Recall that we are interested in the maximal root of
\[
  \det(B_n-\lambda A_n)=c_n \, \lambda \, F_n(\lambda) \, F_{n-1}(\lambda).
\]
We will prove that the maximal root of $\det(B_n-\lambda A_n)$
is equal to the maximal root of $F_n(\lambda)$. We prove this in
Lemma~\ref{mxroot} which is based on Lemmas \ref{comparison},
\ref{boundlemma}, and~\ref{lowerbound}, which are technical in nature. A lower
and an upper bound on the maximal root of $F_n(\lambda)$ are given in
Lemma~\ref{lowerbound}. A better upper bound is given in
Lemma~\ref{upperbound}. These two lemmas are based on Lemma~\ref{boundlemma}.
Recall the definition of $\nu(n)=\lfloor\frac{n}{2}\rfloor$.

\begin{defn}
To simplify notation in this section, we introduce the
polynomials
\begin{equation}\label{eq.fj}
  f_j(n) := \frac{(n-2j+1)_{4j}}{4^j(2j)!},
\end{equation}
which correspond (up to sign) to the coefficients of $F_n(\lambda)$:
\[
  F_n(\lambda) = \sum_{j=0}^{\nu(n)} (-1)^j f_j(n) \lambda^{\nu(n)-j} =
  \lambda^{\nu(n)} - f_1(n)\lambda^{\nu(n)-1} + f_2(n)\lambda^{\nu(n)-2} - \dots
\]
In particular, we have
\begin{align*}
  f_1(n) &= \frac{(n-1)_4}{8} = \frac{n(n-1)(n+1)(n+2)}{8}, \\
  f_2(n) &= \frac{(n-3)_8}{384}, \\
  f_3(n) &= \frac{(n-5)_{12}}{46080}.
\end{align*}
\end{defn}

\begin{lemma}\label{comparison}
Let $n\in\mathbb{N}$ with $n\geq2$. If $\lambda \in\mathbb{R}$ is a root of $F_{n}$
with $\lambda>\frac{1}{2}f_1(n)$ then $F_{n+1}(\lambda)<0$.
\end{lemma}
\begin{proof}
We distinguish two cases depending on the parity of~$n$.

\emph{Case $n=2k+2$. }  We have that $\nu(n)=k+1$ and $\nu(n+1)=k+1$.
Define 
\[
  G_n(x):=F_{n+1}(x)-F_n(x)=\sum_{j=0}^{k}(-4)^{j-k}\frac{(2k-2j+3)_{2k+2}}{(2j+1)!}(j-k-1)x^j.
\]
Our goal is to show that $G_n(x)<0$ for $x>\frac{1}{2}f_1(n)$; the claim then
follows immediately by using the assumption that $\lambda$ is a root of $F_n$.
For this purpose we define $g_j(n)$ to be the absolute value of the
coefficient of $x^j$ in $G_n(x)$ so that
\[
  g_j(n) = 4^{j-k}\frac{(2k-2j+3)_{2k+2}}{(2j+1)!}(k-j+1).
\]
We now want to prove that $\lambda\, g_j(n) >g_{j-1}(n)$ for $1\leq j\leq k$ and
$\lambda>\frac12f_1(n)$, which is implied by
\[
  \frac12f_1(n)\,g_j(n)>g_{j-1}(n),\qquad (1\leq j\leq k).
\]
Substituting for $g_j(n)$ we obtain
\[
  \frac12f_1(n)4^{j-k}\,\frac{(k-j+1)(2k-2j+3)_{2k+2}}{(2j+1)!}>4^{j-1-k}\frac{(k-j+2)(2k-2j+5)_{2k+2}}{(2j-1)!}.
\]
Multiplying this inequality by $(2j-1)!$ and dividing by $4^{j-1-k}(2k-2j+5)_{2k}$,
we obtain
\[
  2f_1(n)\,\frac{(k-j+1)(2k-2j+3)(2k-2j+4)}{2j(2j+1)}>(k-j+2)(4k-2j+5)(4k-2j+6).
\]
Plugging in~$f_1(n)=\frac18(2k+1)(2k+2)(2k+3)(2k+4)$ and substituting $j\to k-j$ leads to
\begin{multline*}
(16j^3+72j^2+88j+16)k^4+(80j^3+360j^2+424j+48)k^3\\
+(172j^3+774j^2+906j+92)k^2+(196j^3+882j^2+1070j+180)k\\
-16j^5-112j^4-212j^3+16j^2+276j+72>0
\end{multline*}
for $0\leq j\leq k-1$. Since $k>j$, the above inequality is true if it is true
for $k=j$. Substituting $k=j$ yields
\[
  16j^7+152j^6+604j^5+1298j^4+1624j^3+1178j^2+456j+72>0,
\]
which is obviously true for all $j\geq 0$. 
Now note that $G_n(\lambda)=F_{n+1}(\lambda)$ because $F_n(\lambda)=0$
by our assumption on~$\lambda$. Finally note that if $k$ is even then 
\[
  G_n(\lambda) = \underbrace{-g_{0\vphantom{j}}}_{\textstyle<0} +
  \sum_{j=1}^{k/2} \underbrace{\bigl(-g_{2j}(n)\lambda+g_{2j-1}(n)\bigr)}_{\textstyle<0}\lambda^{2j-1} < 0,
\]
and if $k$ is odd then
\[
  G_n(\lambda)=\sum_{j=0}^{(k-1)/2}\underbrace{\bigl(-g_{2j+1}(n)\lambda+g_{2j}(n)\bigr)}_{\textstyle<0}\lambda^{2j} < 0.
\]

\emph{Case $n=2k+1$. }  We have that $\nu(n)=k$ and $\nu(n+1)=k+1$.
This time let
\[
  G_n(x):=F_{n+1}(x)-xF_n(x)=
  \sum_{j=0}^{k} (-1)^{j-k+1} \underbrace{4^{j-k} \frac{(2k-2j+3)_{2k+1}}{(2j)!}(k-j+1)}_{\textstyle =:g_j(n)}x^j,
\]
and denote by $g_j(n)$ the absolute value of the coefficient of $x^j$ in
$G_n$, as before.  Again, we want to prove that $G_n(x)<0$ for
$x>\frac{1}{2}f_1(n)$, but we don't want to repeat all the arguments from the
first case. Instead we discuss how this proof can be supported by computer
algebra techniques. First we want to point out that $G_n(x)$, being defined as
a hypergeometric sum, satisfies a linear recurrence equation, and that
inequalities involving such quantities can be proven
algorithmically~\cite{GerholdKauers05}.  However, our experiments suggest that
the present example is computationally too expensive, and therefore we let
computer algebra enter at a later stage of the proof.  In order to prove that
$\lambda \, g_j(n) >g_{j-1}(n)$ for $\lambda>\frac12f_1(n)$ and $1\leq j\leq
k$, we focus on the stronger statement
$\frac12f_1(n)\,g_j(n)/g_{j-1}(n)>1$. Elementary calculations exploiting the
hypergeometric nature of~$g_j(n)$ that are analogous to the previous case lead
to the rational function inequality
\[
  \frac{k(2k+1)(2k+2)(2k+3)(k-j+1)(2k-2j+3)(k-j+2)}{4j(2j-1)(k-j+2)(2k-j+2)(4k-2j+5)}>1.
\]
Now we employ cylindrical algebraic decomposition~\cite{Collins75} to
establish the correctness of the previous inequality: naming it \texttt{ineq},
the Mathematica command\medskip\\ \centerline{
  \texttt{CylindricalDecomposition[Implies[1 <= j <= k, ineq], \string{j, k\string}]}
}\medskip\\ yields \texttt{True} in a fraction of a second~\cite{elec}.

By the assumption on~$\lambda$ we have that $G_n(\lambda)=F_{n+1}(\lambda)$ because $\lambda F_n(\lambda)=0$. 
The proof is concluded by noting that if $k$ is even then 
\[
  G_n(\lambda) = \underbrace{-g_{0\vphantom{j}}}_{\textstyle<0} +
  \sum_{j=1}^{k/2} \underbrace{\bigl(-g_{2j}(n)\lambda+g_{2j-1}(n)\bigr)}_{\textstyle<0}\lambda^{2j-1} < 0
\]
and if $k$ is odd then
\[
  G_n(\lambda)=\sum_{j=0}^{(k-1)/2}\underbrace{\bigl(-g_{2j+1}(n)\lambda+g_{2j}(n)\bigr)}_{\textstyle<0}\lambda^{2j} < 0.
\]
\end{proof}

\begin{lemma}\label{boundlemma}
If $\lambda>\frac12f_1(n)$, then $\lambda\,f_j(n)>f_{j+1}(n)$ for $1\leq j\leq\nu(n)-1$.
\end{lemma}
\begin{proof}
The problem is equivalent to proving
$\frac12f_1(n)f_j(n)/f_{j+1}(n)>1$.
Substituting~\eqref{eq.fj} for $f_j(n)$ we obtain
\[
  \frac{(j+1)(2j+1)(n-1)n(n+1)(n+2)}{2(n-2j-1)(n-2j)(n+2j+1)(n+2j+2)}>1.
\]
Again, this can be proven routinely using cylindrical algebraic
decomposition~\cite{elec}. Alternatively, we clear denominators and collect terms:
\begin{multline}\label{eq.ineq3}
  \left(2 j^2+3 j-1\right) n^4+2 \left(2 j^2+3 j-1\right) n^3+\left(14 j^2+13 j+1\right) n^2+{}\\
  2 (2 j+1) (3 j+1) n-8 j (j+1) (2 j+1)^2 > 0.
\end{multline}
Now observe that $1\leq j\leq\nu(n)-1$ implies $n>2j+2$ and that~\eqref{eq.ineq3}
holds for all $n>2j+2$ if one can show that it holds for $n=2j+2$.
Substituting $n=2j+2$ into~\eqref{eq.ineq3} gives
\[
  32 j^6+208 j^5+536 j^4+700 j^3+424 j^2+60 j-24 > 0
\]
which is clearly true for all $j\geq1$.  
\end{proof}

\begin{defn}
For $n\geq 2$ we define $\lambda_n$ to be the maximal root of~$F_n(\lambda)$.
\end{defn}
We are now ready to give an upper and a lower bound for~$\lambda_n$.

\begin{lemma}\label{lowerbound}
For $n\geq 2$ the maximal root $\lambda_n$ satisfies $m(n)\leq\lambda_n\leq f_1(n)$ with
\begin{align*}
  m(n) &:= \frac{f_1(n)}{2}+\sqrt{\frac{f_1(n)^2}{4}-f_2(n)}\\
       &\phantom{:}= \frac{f_1(n)}{2}\left(1+\sqrt{1-\frac{2}{3}\frac{(n-2)(n-3)(n+3)(n+4)}{n(n-1)(n+1)(n+2)}}\right).
\end{align*}
Moreover, $\lambda_n<f_1(n)$ for $n\geq4$ and $m(n)<\lambda_n$ for $n\geq6$.
\end{lemma}
\begin{proof}
If $n\in\{2,3\}$ then obviously $m(n)=\lambda_n=f_1(n)$ holds.
So let now $n\geq4$ be fixed and set $\nu:=\nu(n)$, i.e., $\nu\geq2$.
In Lemma~\ref{boundlemma} we proved that if $\lambda>\frac12f_1(n)$ then
$\lambda\,f_j(n)>f_{j+1}(n)$. Consequently, under this assumption
on~$\lambda$, we get: if $\nu$ is even then
\[
  \sum_{j=2}^{\nu}(-1)^{j}f_j(n)\lambda^{\nu-j} =
  \sum_{k=1}^{\nu/2-1} \!\! \underbrace{\bigl(\lambda\,f_{2k}(n)-f_{2k+1}(n)\bigr)\lambda^{\nu-2k-1}}_{\textstyle>0} \,
  + \underbrace{f_\nu(n)}_{\textstyle>0}>0,
\]
and if $\nu$ is odd then
\[
  \sum_{j=2}^\nu (-1)^j f_j(n) \lambda^{\nu-j} =
  \sum_{k=1}^{(\nu-1)/2} \!\! \underbrace{\bigl(\lambda\,f_{2k}(n)-f_{2k+1}(n)\bigr)\lambda^{\nu-2k-1}}_{\textstyle>0} > 0.
\]
In particular let now $\lambda\geq f_1(n)$. Then 
\[
  F_n(\lambda) = \underbrace{\lambda^{\nu}-f_1(n)\lambda^{\nu-1}}_{\textstyle\geq 0}
  + \sum_{j=2}^\nu (-1)^j f_j(n) \lambda^{\nu-j} >0.
\]
Therefore the maximal root of $F_n(\lambda)$ cannot exceed $f_1(n)$,
which proves the upper bound. Analogously one finds that
\[
  \sum_{j=3}^\nu (-1)^j f_j(n) \lambda^{\nu-j} \leq 0,
\]
which is strict for all $n\geq6$.
Then for $\lambda = \frac12f_1(n)+\sqrt{\frac14f_1(n)^2-f_2(n)} > \frac12f_1(n)$ we have
\[
  F_n(\lambda) = \underbrace{\lambda^{\nu}-f_1(n)\lambda^{\nu-1}+f_2(n)\lambda^{\nu-2}}_{\textstyle=0}
  + \sum_{j=3}^\nu (-1)^j f_j(n) \lambda^{\nu-j} \leq 0.
\]
Since $F_n(\lambda)\leq0$ and $\lim_{x \to \infty}F_n(x)=+\infty$, the polynomial $F_n(x)$ has a root for $x\geq\lambda$.
This proves the lower bound.
\end{proof}

\begin{lemma}\label{mxroot}
Let $n\geq2$. Then $\lambda_{n+1}>\lambda_n$. 
\end{lemma}
\begin{proof}
By Lemma~\ref{lowerbound} we have that
$\lambda_n>\frac12f_1(n)$.  Then by Lemma~\ref{comparison} we
have that $F_{n+1}(\lambda_n)<0$. Since by definition
$\lim_{x\to\infty}F_n(x)=+\infty$, it follows that between $\lambda_n$ and
$+\infty$ the function $F_{n+1}(x)$ takes the value~$0$ at some
point~$x_0$. In particular $\lambda_n<x_0\leq\lambda_{n+1}$.
\end{proof} 
\begin{cor}\label{maxcor}
The maximal root of $\det(B_{n}-\lambda A_{n})$ is equal to the maximal root of $F_{n}(\lambda)$. 
\end{cor}

\begin{lemma}\label{upperbound}
For $n\geq2$ the maximal root $\lambda_n$ satisfies
$\lambda_n\leq M(n)$ with
\[
  M(n) := \frac{f_1(n)}{3} + \biggl(f_1(n)\left(p_1(n)+\sqrt{p_2(n)}\right)\biggr)^{\!1/3}
   + \biggl(f_1(n)\left(p_1(n)-\sqrt{p_2(n)}\right)\biggr)^{\!1/3}
\]
where the polynomials $p_1$ and $p_2$ are given by
\begin{align*}
  p_1(n) &:= \frac{1}{4320}\bigl(n^8+4 n^7+8 n^6+10 n^5+404 n^4+796 n^3-4733 n^2-5130 n+16200\bigr),\\
  p_2(n) &:= \frac{1}{597196800} (n-3) (n-2) (n+3) (n+4) \bigl(7n^{12}+42n^{11}-641n^{10}-{}\\
  &\qquad 3590n^9-2951n^8+10198n^7-20619n^6-113090n^5+4705644n^4+{}\\
  &\qquad 9619080n^3-40140000n^2-44971200n+116640000\bigr).
\end{align*}
Moreover, we have $\lambda_n<M(n)$ if $n\geq8$.
\end{lemma}
\begin{proof}
The equality $\lambda_n=M(n)$ is easily established for $2\leq n\leq5$,
by using the symbolic simplification capabilities of Mathematica. Next,
for $n\geq 6$ we may write $F_n(\lambda)$ as
\begin{equation}\label{ineqq}
  F_n(\lambda) = \lambda^{\nu(n)-3}\bigl(\lambda^3-f_1(n)\lambda^2+f_2(n)\lambda-f_3(n)\bigr)
  +\sum_{j=4}^{\nu(n)} (-1)^j f_j(n) \lambda^{\nu(n)-j}.
\end{equation}
Obviously, the sum in~\eqref{ineqq} equals zero if $n\in\{6,7\}$ since then $\nu(n)=3$.
By a similar argument as in the proof of Lemma~\ref{lowerbound}, one
sees that this sum is strictly positive for $n\geq8$, provided that $\lambda>\frac12f_1(n)$.
Note that the maximal root of the polynomial 
\begin{equation}\label{deg3}
\lambda^{3}-f_1(n)\lambda^{2}+f_2(n)\lambda-f_3(n)
\end{equation}
is greater than $\frac12f_1(n)$ because the lower bound is the same as the one
derived in Lemma~\ref{lowerbound} using the same arguments as in its proof.
The roots of this third-degree polynomial can be computed by using
Cardano's formulas, namely we want to solve $x^3+bx^2+cx+d=0$. The roots of this
polynomial are given by $y_i-b/3$ for $i=1,2,3$ where
\[
  y_1 := \alpha+\beta \quad\text{and}\quad
  y_{2,3} := -\frac{\alpha+\beta}{2}\pm i\frac{\alpha-\beta}{2}\sqrt{3},
\]
where
\[
  \alpha := \left(-\frac{Q}{2}+\sqrt{\Delta}\right)^{\!1/3}
  \quad\text{and}\quad
  \beta := \left(-\frac{Q}{2}-\sqrt{\Delta}\right)^{\!1/3},
\]
where
\[
  Q := \frac{2b^3}{27}-\frac{bc}{3}+d
  \quad\text{and}\quad
  \Delta := \left(\frac19\bigl(3c-b^2\bigr)\right)^{\!3}+\left(\frac{Q}{2}\right)^{\!2}.
\]
Setting $\lambda=x$, $b=-f_1(n)$, $c=f_2(n)$ and $d=-f_3(n)$ we obtain
$\lambda_i=y_i-b/3$ as the roots of~\eqref{deg3}.

We obtain three real roots when $\Delta<0$ and when $\Delta>0$ we have only
one real root. The latter case happens for $n\geq 10$ and the real root is
$y_1-b/3$. For the cases $n=6,7,8,9$ we have three real roots and one can
check numerically that the maximal root is still $y_1-b/3$. By straight-forward calculations one can verify that
$y_1-b/3=M(n)$ and we get $F_n(M(n))=0$ for $n\in\{6,7\}$. For $n\geq8$ and
$\lambda\geq M(n)$ we have that~\eqref{deg3} is nonnegative which together with
the statements about the sum in~\eqref{ineqq} implies that $F_n(\lambda)>0$.
\end{proof}

Since we have proven that for $n\geq 6$ we have $m(n)<\lambda_{n}<M(n)$ it
follows (from dividing the inequality by $f_1(n)$ and taking the limit
$n\to\infty$) that
\[
  \underbrace{\vphantom{\bigg)}\frac{1}{2}+\frac{1}{2}\sqrt{\frac{1}{3}}}_{\textstyle\sim 0.789} \leq
  \lim_{n\to \infty}\frac{\lambda_n}{f_1(n)}\leq
  \underbrace{\frac{1}{3}+\biggl(\frac{2}{135}+\sqrt{\frac{7}{145800}}\biggr)^{\!1/3}+
    \biggl(\frac{2}{135}-\sqrt{\frac{7}{145800}}\biggr)^{\!1/3}}_{\textstyle\sim 0.811}.
\]
The previous lemmas indicate how to obtain a sequence of better and better
bounds for~$\lambda_n$: while in Lemma~\ref{lowerbound} the root of the
polynomial given by the first three terms of~$F_n(\lambda)$ yields a lower
bound, Lemma~\ref{upperbound} gives an upper bound by considering the first
four terms. A more accurate lower bound would follow from taking the first
five terms, then a better upper bound from the first six terms, etc.

\section{Asymptotic Behaviour of the Roots}
\label{section_asympt}

Since the matrices $M_n$ and $K_n$ defined in~\eqref{eq.M}--\eqref{eq.K} are
symmetric, it follows that the polynomials $F_n(\lambda)$ defined
in~\eqref{eq.f} have only real roots, all of which are positive because the
coefficients of $F_n(\lambda)$ are alternating. When we plot the roots for
different $n\in\mathbb{N}$ we get a very interesting picture, see
Figure~\ref{fig.roots}. Moreover, one sees that the smallest root of
$F_{2n}(\lambda)$ converges to a specific value as $n$ goes to infinity, and
the same is true for the smallest root of $F_{2n+1}(\lambda)$. The situation
is similar when considering the second-smallest root, the third-smallest root,
and so on. The following proposition makes this observation precise.

\begin{figure}[t]
\includegraphics[width=0.7\textwidth]{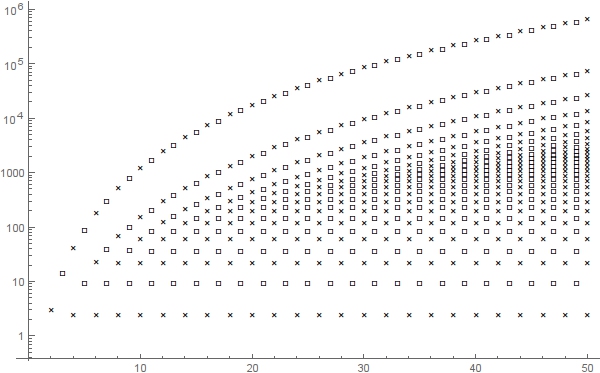}
\caption{Distribution of the roots of $F_n(\lambda)$ for $2\leq n\leq 50$ on a
logarithmic scale; for even~$n$ the locations of the roots are marked by crosses,
for odd~$n$ with squares.}
\label{fig.roots}
\end{figure}

\begin{prop}
Let $F_n(\lambda)$ be defined as in~\eqref{eq.f} and
let $\lambda^{(0)}_{n,1}<\dots<\lambda^{(0)}_{n,n}$ denote the roots of
$F_{2n}(\lambda)$ in increasing order, and similarly
$\lambda^{(1)}_{n,1}<\dots<\lambda^{(1)}_{n,n}$ denote the roots of
$F_{2n+1}(\lambda)$. Then for fixed $k\in\mathbb{N}$ we have
\[
  \lim_{n\to\infty} \lambda^{(0)}_{n,k} = \bigl(k-\tfrac12\bigr)^2\pi^2
  \qquad\text{and}\qquad
  \lim_{n\to\infty} \lambda^{(1)}_{n,k} = k^2\pi^2.
\]
\end{prop}
\begin{proof}
The coefficient of $\lambda^j$ in $F_{2n}(\lambda)$ is, according
to~\eqref{eq.f}, given by
\[
  \frac{(-4)^{j-n} (2n-2j+1)_{2n}}{(2j)!}.
\]
We normalize the monic polynomials $F_{2n}$ such that their constant
coefficient is~$1$, i.e., we divide $F_{2n}$ by $(-4)^{-n} (2n+1)_{2n}$, and
obtain for the coefficient of $\lambda^j$ in these normalized polynomials:
\[
  \frac{(-4)^j (2n-2j+1)_{2n}}{(2j)! (2n+1)_{2n}}
  = \frac{(-1)^j}{(2j)!} \cdot \frac{4^j (2n-2j+1)_{2j}}{(4n-2j+1)_{2j}}.
\]
Obviously the second factor is, for fixed~$j$, a rational function in~$n$ with
numerator and denominator having the same degree~$2j$ and the same leading
coefficient~$16^j$; hence it tends to~$1$ as~$n$ goes to infinity. This means
that the power series obtained as the limit of the normalized polynomials is
\[
  \sum_{j=0}^\infty \frac{(-1)^j}{(2j)!}x^j = \cos(\sqrt{x})
\]
whose roots are precisely the limiting values in the assertion.  The limit of
$F_{2n+1}(\lambda)$ can be computed analogously and yields the Taylor
expansion of $\sin(\sqrt{x})/\sqrt{x}$.
\end{proof}

Recall that we are actually not interested in the smallest root of
$F_n(\lambda)$ but in the largest one. Its asymptotic behaviour can be
extracted in a similar fashion.

\begin{prop}
Let $F_n(\lambda)$ be defined as in~\eqref{eq.f} and let $\lambda_n$
denote the largest root of $F_n$, as before. Then
\[
  \lim_{n\to\infty}\frac{\lambda_n}{n^4} = \frac{1}{\pi^2}.
\]
\end{prop}
\begin{proof}
Let $\hat{F}_n(\lambda)$ denote the reciprocal polynomial of $F_n(\lambda)$,
which means that $\hat{F}_n(\lambda)=\lambda^{\nu(n)}F_n(1/\lambda)$ where
$\nu(n)=\lfloor n/2\rfloor$ is the degree of~$F_n$. Then the largest root of
$F_n$ equals the reciprocal of the smallest root of $\hat{F}_n$. Now consider
the family of polynomials
\[
  \hat{F}_n\left(\frac{\lambda}{n^4}\right) = \sum_{j=0}^{\nu(n)} \frac{(2j+1)_n}{(-4n^4)^j \, (n-2j)!} \lambda^j.
\]
The coefficient of $\lambda^j$ in these polynomials tends to $(-4)^{-j}/(2j)!$
as $n$ goes to infinity. Hence in the limit we obtain the power series
\[
  \sum_{j=0}^\infty \frac{x^j}{(-4)^j \, (2j)!} = \cos\left(\frac{\sqrt{x}}{2}\right),
\]
whose smallest root is $\pi^2$. The claim follows.
\end{proof}

Note that this result is in accordance with the bounds derived in
Section~\ref{section_bounds}, in particular with the inequality stated at the
end of that section: the numerical value of $8\pi^{-2}$ is approximately
$0.810569$ which is very close to the previously derived upper bound. The
reason why the upper bound is more accurate comes from the fact that a
third-degree approximation of $F_n(\lambda)$ was taken (in
Lemma~\ref{upperbound}), whereas the lower bound was obtained from a
second-degree polynomial (see Lemma~\ref{lowerbound}).

\section{The Boundary Estimate}\label{section_easydet}

Finally we tackle the second kind of problem, corresponding to Equation~\eqref{equ_referenceInverseInequality2}.
In this instance it is advantageous to formulate it using the Legendre basis. Thus we have to solve
the eigenvalue problem
\[
  L_n \vec{x}_n = \mu_n M_n \vec{x}_n
\]
with the following $n\times n$ matrices $L_n$ and $M_n$: the $(i,j)$ entry
of $L_n$ is given by
\[
  P_i(1)P_j(1) + P_i(-1)P_j(-1)
\]
whereas in $M_n$ one has
\[
  \int_{-1}^1 P_i(x)P_j(x)\,\mathrm{d}x.
\]
Here the basis functions are the venerable Legendre
polynomials~$P_n(x)$. Taking into account the well-known evaluations
$P_n(1)=1$ and $P_n(-1)=(-1)^n$ this is equivalent to finding the roots of the
determinant of $C_n=(c_{i,j})_{1\leq i,j\leq n}$ whose matrix entries are
given by
\[
  c_{i,j} := 1+(-1)^{i+j}-\delta_{i,j}\frac{2\mu}{2i+1}.
\]

Obviously the matrix $C_n$ has zeros at all positions $(i,j)$ for which $i+j$
is an odd integer. As in Section~\ref{section_detEvaluation} we decompose it
into block form and obtain
\[
  \det(C_n) =
  \begin{vmatrix}C^{(0)}_{\lfloor n/2\rfloor} & 0\\ 0& C^{(1)}_{\lceil n/2\rceil}\end{vmatrix} =
  \det\Bigl(C^{(0)}_{\lfloor n/2\rfloor}\Bigr) \cdot \det\Bigl(C^{(1)}_{\lceil n/2\rceil}\Bigr)
\]
where the subscripts indicate the dimension of the square matrices $C^{(0)}$ and $C^{(1)}$,
whose entries are independent of the dimension and given by
\[
  c^{(0)}_{i,j} := 2 - \delta_{i,j}\frac{2\mu}{4i+1}
  \quad\text{and}\quad
  c^{(1)}_{i,j} := 2 - \delta_{i,j}\frac{2\mu}{4i-1}.
\]
\begin{thm}
For all nonnegative integers~$n$ we have
\begin{align*}
\det\bigl(C^{(0)}_n\bigr) &= \frac{(-1)^n}{2^n\left(\frac54\right)_n}\mu^{n-1}\bigl(\mu-2n^2-3n\bigr),\\
\det\bigl(C^{(1)}_n\bigr) &= \frac{(-1)^n}{2^n\left(\frac34\right)_n}\mu^{n-1}\bigl(\mu-2n^2-n\bigr).
\end{align*}
\end{thm}
\begin{proof}
By some elementary row operations, the matrix $C^{(0)}_n$ is brought to
triangular form. First we subtract the first row from rows $2$ through~$n$,
obtaining the following matrix: the $(1,1)$ entry is $2-\frac{2\mu}{5}$,
the remaining entries in the first row are~$2$, the remaining entries of the
first column are $\frac{2\mu}{5}$, and the diagonal entries $(i,i)$ are
$-\frac{2\mu}{4i+1}$ for $i>1$; the rest are zeros. So in order to
transform the matrix to lower triangular form, we multiply row~$i$, for $2\leq
i\leq n$, by $\frac{4i+1}{\mu}$ and add it to the first row. Thus the
$(1,1)$-entry becomes
\[
  2-\frac{2\mu}{5} + \sum_{i=2}^n \frac{2\mu}{5}\frac{4i+1}{\mu} = \frac25\bigl(2n^2+3n-\mu\bigr).
\]
It follows that the determinant of $C^{(0)}_n$ is
\[
  \frac25\bigl(2n^2+3n-\mu\bigr)\prod_{i=2}^n \frac{-2\mu}{4i+1} =
  \frac{(-1)^n}{2^n\left(\frac54\right)_n}\mu^{n-1}\bigl(\mu-2n^2-3n\bigr),
\]
as claimed. The evaluation of $\det\bigl(C^{(1)}_n\bigr)$ is obtained in a
completely analogous way.
\end{proof}
\begin{cor}
For all nonnegative integers~$n$ we have
\begin{align*}
  \det(C_n) &= \det_{1\leq i,j\leq n} \Bigl(1+(-1)^{i+j}-\delta_{i,j}\frac{2\mu}{2i+1}\Bigr) \\
  &= \frac{(-1)^n}{\left(\frac32\right)_n}\mu^{n-2}
  \Bigl(\mu-2\left\lfloor\frac{n}{2}\right\rfloor^2-3\left\lfloor\frac{n}{2}\right\rfloor\Bigr)
  \Bigl(\mu-2\left\lceil\frac{n}{2}\right\rceil^2-\left\lceil\frac{n}{2}\right\rceil\Bigr) \\
  &= \frac{(-1)^n}{\left(\frac32\right)_n}\mu^{n-2}
  \begin{cases}
    \bigl(\mu-\frac{n^2+3n}{2}\bigr)\bigl(\mu-\frac{n^2+n}{2}\bigr), & \text{if }n\text{ is even},\\
    \bigl(\mu-\frac{n^2+3n+2}{2}\bigr)\bigl(\mu-\frac{n^2+n-2}{2}\bigr), & \text{if }n\text{ is odd}.
  \end{cases}
\end{align*}
\end{cor}

The previous corollary now gives an answer to the original eigenvalue problem, namely
that the largest eigenvalue $\mu_n$ of $L_n \vec{x}_n = \mu_n M_n \vec{x}_n$ is
\[
  \mu_n = \begin{cases}
  \frac12n(n+3) & \text{if }n\text{ is even},\\
  \frac12n(n+3)+1 & \text{if }n\text{ is odd}.
  \end{cases}
\]

\section{Outlook and future work}
In this work we have presented tools from symbolic computation to give precise
estimates for two types of inverse inequalities. It would be interesting to
apply these methods on other types of elements, like simplices for
example. Moreover, in Isogeometric Analysis (IgA) the constants in the inverse
inequalities depend on three parameters, i.e. the mesh size, the polynomial
degree and the smoothness factor. For this case not so much is known and it
would be attractive to apply tools from symbolic computation also in this
case.
  
\subsection*{Acknowledgment}
We are grateful to Ulrich Langer for initiating this collaboration between the
Numerical Analysis Group and the Symbolic Computation Groups at JKU, RISC, and
RICAM. We want to thank Peter Paule for insightful discussions on the
properties of the polynomials~$F_n(\lambda)$. We also thank both of them for
constantly motivating us to put a serious effort into solving this important
problem. Last but not least we appreciate the careful reading and the helpful
comments of the anonymous referee.

\bibliographystyle{plain}

\end{document}